\documentclass[11pt]{article}
\usepackage{amsfonts}
\usepackage{mathrsfs}
\usepackage{amsthm}
\usepackage{amssymb}
\usepackage{amsmath}
\usepackage{enumerate}
\usepackage{colortbl}
\usepackage{subfigure}
\usepackage{graphicx}
\usepackage{geometry}

\theoremstyle{plain}
\newtheorem{theorem}{Theorem}[section]
\newtheorem{proposition}[theorem]{Proposition}

\newtheorem{lemma}[theorem]{Lemma}

\theoremstyle{definition}
\newtheorem{definition}{Definition}[section]

\theoremstyle{remark}
\newtheorem{remark}{\textbf{Remark}}[section]

\theoremstyle{example}

\numberwithin{equation}{section}

\title{Stationary Measures of Space-Inhomogeneous Three-State Quantum Walks on Line: Revisited}

\author{Shengsheng Liu,\ \  Caishi Wang,\ \ Jijun Zhao\\
    School of Mathematics and Statistics, Northwest Normal University\\
    Lanzhou, Gansu 730070, People's Republic of China}

\begin{document}
\maketitle

\noindent\textbf{Abstract.}\ \
Of a quantum walk, its stationary measures play an important role in understanding its evolution behavior.
In this paper we investigate stationary measures of two models of space-inhomogeneous three-state quantum walk on the line.
By using the method of reduced matrix, we find out stationary measures of the two models under some mild conditions. Our results
generalize the corresponding ones existing in the literature.
\vskip 2mm

\vskip 2mm


\section{Introduction}\label{sec-1}

Quantum walks (also known as quantum random walks) are quantum analogs of classical random walks in probability,
and have found successful applications in quantum computing and simulation of physical processes (see, e.g., \cite{venegas,bose, cw,zhan} and references therein).
Due to the quantum interference effects, quantum walks greatly outperform classical random walks at certain computational tasks,
and moreover it has turned out that quantum walks constitute universal models of quantum computation.
There are two basic categories of quantum walks: discrete-time ones and continuous-times ones. In this paper we only focus
on the discrete-time ones, which we simply call quantum walks below.

Three-state quantum walks on the line play an important role in understanding the evolution behavior of a general quantum walk.
In the past two decades, many researches have been done on limit probability distributions and localization of these walks.
Recent years have seen much interest in stationary measures of these walks, which can help understand evolution behavior of these walks.
Konno \cite{konno} analyzed the Grover walk, which is a space-homogeneous three-state walk on the line, and found that its stationary measure
decays exponentially with respect to position.
In 2015, Wang et al. \cite{wlw} considered a space-inhomogeneous three-state walk on the line, called the three-state Wojcik walk,
whose time evolution is determined by unitary matrices $U_x = e^{i v_{x}}G$, $x\in \mathbb{Z}$, where $v_x$ is defined by
\begin{equation*}
  v_{x}=
\left\{
  \begin{array}{ll}
    2 \pi \theta, & \hbox{$x=0$;} \\
    0, & \hbox{$x\ne 0$, $x\in \mathbb{Z}$.}
  \end{array}
\right.
\end{equation*}
with $\theta\in(0,1)$ and $G$ is the Grover matrix, namely
\begin{equation}\label{eq-1-1}
G=\frac{1}{3}\left(\begin{array}{ccc}
-1 & 2 & 2 \\
2 & -1 & 2 \\
2 & 2 & -1
\end{array}\right)
\end{equation}
By using the splitted generating function (SGF) method, they obtained the stationary measures of the three-state Wojcik walk \cite{wlw}.
In 2016, Endo et al. \cite{ekk-1} investigated the same walk and found a relation between the stationary and the limit measures of the walk.
In 2017, Kawai et al. \cite{kkk} introduced a novel method to deal with space-homogeneous three-state quantum walks, and two years later by using the method
Han et al. \cite{hgyc} calculated stationary measures of the Grover walk with one defect, which belongs to the category of space-inhomogeneous three-state quantum walks.
Nowadays the method introduced in \cite{kkk} is known as the method of reduced matrix in the literature.

In this paper, we would like to show that the method of reduced matrix can be also applied to more general space-inhomogeneous three-state quantum walks.
More precisely, we would like to use the method of reduced matrix to deal with two models of space-inhomogeneous three-state quantum walk on the line
(see Section~\ref{sec-3} for their detailed descriptions), which include the quantum walks considered in \cite{wlw, ekkt} as a special case.

The paper is organized as follows. In Section~\ref{sec-2}, we give the definition of a general space-inhomogeneous three-state quantum walk
on the line and state some known results about the eigenvalue of the evolution operator of a space-homogeneous three-state quantum walk on the line.
Our main work lies in Section~\ref{sec-3}, where we first describe our two models of space-inhomogeneous three-state quantum walk and then
calculate their stationary measures. Finally in Section~\ref{sec-4}, we make some conclusion remarks.

Throughout this paper $\mathbb{Z}$ denotes the integer lattice, while $\mathbb{C}$ denotes the complex numbers. By convention $\mathbb{C}^3$ means the $3$-dimensional complex Euclidean space.

\section{Preliminaries}\label{sec-2}

In this section, we recall some necessary notions, notation and facts about three-state quantum walks on the line.

Let $l^2(\mathbb{Z}, \mathbb{C}^3)$ be the space of all square summable $\mathbb{C}^3$-valued functions defined on $\mathbb{Z}$, namely
\begin{equation*}
  l^2(\mathbb{Z}, \mathbb{C}^3) = \left\{\Psi \colon \mathbb{Z}\rightarrow \mathbb{C}^3 \,\biggm|\, \sum_{x\in \mathbb{Z}} \|\Psi(x)\|^2 <\infty \right\},
\end{equation*}
where $\|\cdot\|$ stands for the norm in $\mathbb{C}^3$. Note that  $l^2(\mathbb{Z}, \mathbb{C}^3)$ is a complex Hilbert of infinite dimension
with a countably-infinite orthonormal basis $\{\phi_z \mid z\in \mathbb{Z}\}$, where $\phi_z$ is the function on $\mathbb{Z}$ given by
$\phi_z(z)=1$ and $\phi_z(x)=0$ for all $x\in \mathbb{Z}$ with $x\ne z$.

Let $\mathfrak{C} = \{C_{x} \mid x\in \mathbb{Z}\}$ be a family of $3\times 3$ unitary matrices indexed by $\mathbb{Z}$, where $C_{x}$ has entries of the following form
\begin{equation}\label{eq-2-1}
  C_{x}
  =\left(\begin{array}{lll}
   a_{x} & b_{x} & c_{x} \\
   d_{x} & e_{x} & f_{x} \\
   g_{x} & h_{x} & k_{x}
\end{array}\right).
\end{equation}
Then each $C_{x}$ has a decomposition of the form  $C_{x} = C_x^{(u)} + C_x^{(m)} + C_x^{(l)}$, where
\begin{equation}\label{eq-2-2}
\begin{aligned}
C_{x}^{(u)} &=\left(\begin{array}{ccc}
a_{x} & b_{x} & c_{x} \\
0 & 0 & 0 \\
0 & 0 & 0
\end{array}\right),\ \
C_{x}^{(m)}=\left(\begin{array}{ccc}
0 & 0 & 0 \\
d_{x} & e_{x} & f_{x} \\
0 & 0 & 0
\end{array}\right),\ \
C_{x}^{(l)}=\left(\begin{array}{ccc}
0 & 0 & 0 \\
0 & 0 & 0 \\
g_{x} & h_{x} & k_{x}
\end{array}\right).
\end{aligned}
\end{equation}
It can be shown that there exists a unique unitary operator $\mathcal{U}_{\mathfrak{C}}$ on $l^2(\mathbb{Z}, \mathbb{C}^3)$ such that
\begin{equation}
  [\mathcal{U}_{\mathfrak{C}}\Psi](x) = C_{x-1}^{(u)}\Psi(x-1) + C_{x}^{(m)}\Psi(x) + C_{x+1}^{(l)}\Psi(x+1),\quad x\in \mathbb{Z},
\end{equation}
where $\Psi$ ranges over $l^2(\mathbb{Z}, \mathbb{C}^3)$.

\begin{definition}\label{def-2-1}
The unitary operator $\mathcal{U}_{\mathfrak{C}}$ indicated above is called the unitary operator determined
by the family $\mathfrak{C} =\{C_{x} \mid x\in \mathbb{Z}\}$ of unitary matrices.
\end{definition}

Let $A$ be a fixed $3\times 3$ unitary matrix. Then, by letting $C_{x} = A$ for all $x\in \mathbb{Z}$, one gets a family $\mathfrak{C} =\{C_{x} \mid x\in \mathbb{Z}\}$
of $3\times 3$ unitary matrices indexed by $\mathbb{Z}$ and the unitary operator $\mathcal{U}_{\mathfrak{C}}$ determined by the family
$\mathfrak{C} =\{C_{x} \mid x\in \mathbb{Z}\}$.
In that case, we simply call $\mathcal{U}_{\mathfrak{C}}$ the unitary operator determined by the unitary matrix $A$ and write it as $\mathcal{U}_{A}$ instead.

\begin{definition}\label{def-2-2}
Let $\mathfrak{C} = \{C_{x} \mid x\in \mathbb{Z}\}$ be a family of $3\times 3$ unitary matrices indexed by $\mathbb{Z}$.
The space-inhomogeneous three-state quantum walk on the line (for brevity, the walk below) determined by $\mathfrak{C} = \{C_{x} \mid x\in \mathbb{Z}\}$
is the one that admits the following features:
\begin{enumerate}
  \item[(1)] the state space of the walk is $l^2(\mathbb{Z},\mathbb{C}^3)$ and its states are represented by unit vectors in $l^2(\mathbb{Z},\mathbb{C}^3)$;
  \item[(2)] the evolution of the walk is governed by equation
             \begin{equation}\label{eq-evolution}
               \Phi_{n} = \mathcal{U}_{\mathfrak{C}}^n \Phi_{0},\quad n\geq 0,
             \end{equation}
             where $\Phi_{n}$ denotes the state of walk at time $n\geq 0$, in particular $\Phi_{0}$ is the initial state, and $\mathcal{U}_{\mathfrak{C}}$ is the
             unitary operator determined by $\mathfrak{C} = \{C_{x} \mid x\in \mathbb{Z}\}$;
  \item[(3)] the quantity $\|\Phi_{n}(x)\|^2$ is the probability of finding the walker at position $x\in \mathbb{Z}$ at time $n\geq 0$.
\end{enumerate}
\end{definition}

Usually, the unitary operator $\mathcal{U}_{\mathfrak{C}}$ indicated in (\ref{eq-evolution}) is known as the evolution operator of the walk,
while the family $\mathfrak{C} = \{C_{x} \mid x\in \mathbb{Z}\}$ is referred to as the coin matrices (of the walk), which describe
the walk's internal degrees of freedom. In the language of physics, the value $\Phi_{n}(x)$ of the state $\Phi_{n}$ at position $x$ is called the
probability amplitude. When the evolution operator is $\mathcal{U}_{A}$, the unitary operator determined by a single unitary matrix $A$,
we say the walk is space-homogeneous.

\begin{remark}\label{rem-2-1}
Let $\mathscr{F}(\mathbb{Z},\mathbb{C}^3)$ be the set of all functions $\Psi\colon \mathbb{Z} \rightarrow\mathbb{C}^3$, which forms a complex linear
space with the usual addition and scalar multiplication and includes $l^2(\mathbb{Z},\mathbb{C}^3)$ as a linear subspace.
Let $\mathcal{U}_{\mathfrak{C}}$ be the unitary operator determined by a family $\mathfrak{C} =\{C_{x} \mid x\in \mathbb{Z}\}$ of unitary matrices.
Then, $\mathcal{U}_{\mathfrak{C}}$ has an extension $\widetilde{\mathcal{U}_{\mathfrak{C}}}$ to $\mathscr{F}(\mathbb{Z},\mathbb{C}^3)$ satisfying
\begin{equation}
  \widetilde{[\mathcal{U}_{\mathfrak{C}}}\Psi](x) = C_{x-1}^{(u)}\Psi(x-1) + C_{x}^{(m)}\Psi(x) + C_{x+1}^{(l)}\Psi(x+1),\quad x\in \mathbb{Z},
\end{equation}
where $\Psi$ ranges over $\mathscr{F}(\mathbb{Z},\mathbb{C}^3)$. In what follows, we call $\widetilde{\mathcal{U}_{\mathfrak{C}}}$ the natural extension
of $\mathcal{U}_{\mathfrak{C}}$ and, instead of $\widetilde{\mathcal{U}_{\mathfrak{C}}}$, we simply use $\mathcal{U}_{\mathfrak{C}}$ to mean the natural extension.
\end{remark}

For $\Psi\in \mathscr{F}(\mathbb{Z},\mathbb{C}^3)$, we denote by $\nu(\Psi)$ the measure on $\mathbb{Z}$ given by
$[\nu(\Psi)](x) = \|\Psi(x)\|^2$, $x\in \mathbb{Z}$.

\begin{definition}
Consider the walk determined by a family $\mathfrak{C} = \{C_{x} \mid x\in \mathbb{Z}\}$ of $3\times 3$-unitary matrices. We define
\begin{equation}\label{eq}
  \mathcal{M}(\mathcal{U}_{\mathfrak{C}})
  = \big\{ \nu(\Psi) \mid \Psi\in \mathscr{F}(\mathbb{Z},\mathbb{C}^3),\, \nu(\mathcal{U}_{\mathfrak{C}}^n\Psi) = \nu(\Psi), \forall\, n\geq 0\big\},
\end{equation}
where the first $\mathcal{U}_{\mathfrak{C}}$ is the evolution operator of the walk and the second $\mathcal{U}_{\mathfrak{C}}$ is the natural extension of the first one.
Elements of $\mathcal{M}(\mathcal{U}_{\mathfrak{C}})$ are called stationary measures of the walk if $\mathcal{M}(\mathcal{U}_{\mathfrak{C}})\neq \emptyset$.
\end{definition}

Consider the walk determined by a family $\mathfrak{C} = \{C_{x} \mid x\in \mathbb{Z}\}$ of $3\times 3$-unitary matrices
and its evolution operator $\mathcal{U}_{\mathfrak{C}}$.
Let $\Psi \in l^2(\mathbb{Z},\mathbb{C}^3)$ with $\Psi\ne 0$ be an eigenvector of $\mathcal{U}_{\mathfrak{C}}$.
Then there exists some $\lambda\in \mathbb{C}$ with $|\lambda|=1$ such that $\mathcal{U}_{\mathfrak{C}}\Psi = \lambda \Psi$. By direct calculation, one gets
\begin{equation*}
  \nu(\mathcal{U}_{\mathfrak{C}}^n\Psi) = \nu(\lambda^n\Psi) = |\lambda^n|^2\nu(\Psi) = |\lambda|^{2n}\nu(\Psi) =\nu(\Psi),\quad \forall\, n\geq 0,
\end{equation*}
hence $\nu(\Psi)$ is a stationary measure of the walk. This observation actually suggests a way to find out a stationary measure of the walk.

Let $\mathcal{U}_A$ be the unitary operator on $l^2(\mathbb{Z},\mathbb{C}^3)$ determined by a $3\times 3$-unitary matric $A$,
where $A$ has entries of the following form
\begin{equation*}
   A
  =\left(\begin{array}{lll}
   a_{11} & a_{12} & a_{13} \\
   a_{21} & a_{22} & a_{23} \\
   a_{31} & a_{32} & a_{33}
\end{array}\right)
\end{equation*}
with $a_{ij}\ne 0$ for all $1\leq i,\, j \leq 3$ and $|a_{22}|\ne 1$. As noted above (see Remark~\ref{rem-2-1}), we still use $\mathcal{U}_A$ to mean the natural extension
of the unitary operator $\mathcal{U}_A$ to $\mathscr{F}(\mathbb{Z},\mathbb{C}^3)$.
Write
$$
B=\left|\begin{array}{ll}
a_{11} & a_{12} \\
a_{21} & a_{22}
\end{array}\right|,
\quad C=\left|\begin{array}{ll}
a_{12} & a_{13} \\
a_{22} & a_{23}
\end{array}\right|, \quad D=\left|\begin{array}{ll}
a_{21} & a_{22} \\
a_{31} & a_{32}
\end{array}\right|,
\quad E=\left|\begin{array}{ll}
a_{22} & a_{23} \\
a_{32} & a_{33}
\end{array}\right|.
$$
Consider the eigenvalue problem $\mathcal{U}_A\Psi = \lambda\Psi$ in $\mathscr{F}(\mathbb{Z},\mathbb{C}^3)$.
The next two lemmas present formulas to construct an eigenvector of $\mathcal{U}_A$.

\begin{lemma}\cite{kkk}\label{lem-2-1}
Assume that $\lambda=-\frac {C} {a_ {13}}=-\frac {D} {a_ {31}} $ with $|\lambda|=1$. Take $\varphi_{1}$, $\varphi_{3}\in \mathbb{C}$
and define $\Psi\colon \mathbb{Z}\rightarrow \mathbb{C}^3$ as
\begin{equation}\label{}
\Psi(x)=\left(\begin{array}{c}
\left(\tilde{a}_{1}^{-1} \lambda\right)^{x} \varphi_{1} \\
-\frac{a_{13}}{a_{12} a_{23}}\left(a_{21}\left(\tilde{a}_{1}^{-1} \lambda\right)^{x} \varphi_{1}+a_{23}\left(\tilde{a}_{2} \lambda^{-1}\right)^{x} \varphi_{3}\right) \\
\left(\tilde{a}_{2} \lambda^{-1}\right)^{x} \varphi_{3}
\end{array}\right),\quad x\in \mathbb{Z},
\end{equation}
where
$$
\tilde{a}_{1}=a_{11}-\frac{a_{13} a_{21}}{a_{23}}, \quad \tilde{a}_{2}=a_{33}-\frac{a_{23} a_{31}}{a_{21}}.
$$
Then it holds that $\mathcal{U}_A\Psi = \lambda\Psi$.
\end{lemma}

\begin{lemma}\cite{kkk}\label{lem-2-2}
Assume that $\lambda=\frac{B}{a_{11}}=\frac{E}{a_{33}}$ with $|\lambda|=1$ and $\lambda^{2}=\tilde{a}_{1} \tilde{a}_{2}$.
Take a function $\varphi\colon \mathbb{Z}\rightarrow \mathbb{C}$ such that $\varphi \ne 0$ (here $0$ means the null function)
and define $\Psi\colon \mathbb{Z}\rightarrow \mathbb{C}^3$ as
\begin{equation}\label{}
\Psi(x)=\left(\begin{array}{c}
\varphi_{x} \\
-\frac{a_{11}}{a_{12} a_{21}}\left\{a_{21} \varphi_{x}+a_{23}\left(\tilde{a}_{1}^{-1} \lambda\right) \varphi_{x-1}\right\} \\
\left(\tilde{a}_{1}^{-1} \lambda\right) \varphi_{x-1}
\end{array}\right),\quad x \in \mathbb{Z}.
\end{equation}
Here$$\tilde{a}_{1}=a_{13}-\frac{a_{11} a_{23}}{a_{21}}, \quad \tilde{a}_{2}=a_{31}-\frac{a_{21} a_{33}}{a_{23}}.$$
Then it holds that $\mathcal{U}_A\Psi = \lambda\Psi$.
\end{lemma}

\section{Main results}\label{sec-3}

In the present section, we state and prove our main results about stationary measures of two models of space-inhomogeneous three-state quantum walk on the line.

\subsection{Model I}

We first describe the model we will deal with in this subsection. Consider the unitary matrices $U_x^{(\phi)} = e^{iv_x}G^{(\phi)}$, $x\in \mathbb{Z}$,
where $v_x$ is defined by
\begin{equation*}
  v_{x}=
\left\{
  \begin{array}{ll}
    2 \pi \theta, & \hbox{$x=0$;} \\
    0, & \hbox{$x\ne 0$, $x\in \mathbb{Z}$}
  \end{array}
\right.
\end{equation*}
with $\theta\in(0,1)$ and $G^{(\phi)}$ is the generalized Grover matrix, namely
\begin{equation}\label{eq-3-1}
G^{(\phi)}
= \frac{1}{3}\left(\begin{array}{ccc}
-\cos \phi & 2 \cos \phi &2 \cos \phi- i 3\sin \phi \\
2 \cos \phi & -\cos \phi- i3 \sin \phi & 2 \cos \phi \\
2 \cos \phi -i 3\sin \phi & 2 \cos \phi & -\cos \phi
\end{array}\right),
\end{equation}
where $\phi$ is a real parameter. Clearly, $G^{(0)}=G$, namely the Grover matrix $G$ is the special case of $G^{(\phi)}$ when $\phi=0$.

According to Definitions~\ref{def-2-1} and \ref{def-2-2}, the family $\mathfrak{U}^{(\phi)}:=\big\{U_x^{(\phi)} = e^{iv_x}G^{(\phi)}\mid x\in \mathbb{Z}\big\}$
of unitary matrices determines a space-inhomogeneous three-state quantum walk on $\mathbb{Z}$ and its evolution operator is the unitary operator
$\mathcal{U}_{\mathfrak{U}^{(\phi)}}$ on $l^2(\mathbb{Z},\mathbb{C}^3)$ determined by the family $\mathfrak{U}^{(\phi)}$.
In what follows, we simply call this model the walk $\mathcal{U}_{\mathfrak{U}^{(\phi)}}$.

\begin{remark}\label{rem-3-1}
The walk $\mathcal{U}_{\mathfrak{U}^{(\phi)}}$ includes those considered in \cite{wlw, ekk-2} as a special case.
\end{remark}

Note that the generalized Grover matrix $G^{(\phi)}$ is also a unitary matrix, hence it determines a space-homogeneous three-state quantum walk on $\mathbb{Z}$,
whose evolution operator is the unitary operator $\mathcal{U}_{G^{(\phi)}}$ determined by $G^{(\phi)}$.
Similarly, we call the space-homogeneous three-state quantum walk determined by $G^{(\phi)}$ the walk $\mathcal{U}_{G^{(\phi)}}$ below.

As can be seen, there are close links between the walk $\mathcal{U}_{\mathfrak{U}^{(\phi)}}$ and the walk $\mathcal{U}_{G^{(\phi)}}$.
The next proposition offers stationary measures of the walk $\mathcal{U}_{G^{(\phi)}}$, which can help find out stationary measures of the walk $\mathcal{U}_{\mathfrak{U}^{(\phi)}}$.

\begin{proposition}\label{prop-3-1}
Consider the walk $\mathcal{U}_{G^{(\phi)}}$.
Let the parameter $\phi$ be such that $\phi\in[0, 2\pi)$ and $\cos \phi \ne 0$. Then, for any function $\varphi\colon \mathbb{Z} \rightarrow \mathbb{C}$
with $\varphi\ne 0$, the walk $\mathcal{U}_{G^{(\phi)}}$ has a corresponding stationary measure of the following form
\begin{equation}\label{eq-3-2}
  \mu(x)=\frac{5}{4}\left(|\varphi(x)|^{2}+ |\varphi(x-1)|^{2}\right)+\frac{1}{2} \Re\left(\varphi(x) \overline{\varphi(x-1)}\right),\quad x\in \mathbb{Z},
\end{equation}
where $\Re(u)$ means the real part of a complex number $u$.
\end{proposition}

\begin{proof}
Let $\varphi\colon \mathbb{Z} \rightarrow \mathbb{C}$ be any function with $\varphi\ne 0$.
Careful calculations give $\frac{B}{a_{11}}=\frac{E}{a_{33}}=\tilde{a}_{1}=\tilde{a}_{2}=e^{-i \phi}$.
Let $\lambda = e^{-i \phi}$. Then $\lambda = \frac{B}{a_{11}}=\frac{E}{a_{33}}$ with $|\lambda|=1$ and $\lambda^{2}=\tilde{a}_{1} \tilde{a}_{2}$.
Thus, by Lemma~\ref{lem-2-2}, $\lambda = e^{-i \phi}$ is an eigenvalue of $\mathcal{U}_{G^{(\phi)}}$ and the corresponding eigenvector $\Psi$ admits
a representation of the form
\begin{align*}
\begin{array}{c}
\Psi(x)=\left(\begin{array}{c}
\varphi_{x} \\
-\frac{a_{11}}{a_{12} a_{21}}\left\{a_{21} \varphi_{x}+a_{23}\left(\tilde{a}_{1}^{-1} \lambda\right) \varphi_{x-1}\right\} \\
\left(\tilde{a}_{1}^{-1} \lambda\right) \varphi_{x-1}
\end{array}\right)
=\left(\begin{array}{c}
\varphi_{x} \\
\frac{1}{2}\left(\varphi_{x}+\varphi_{x-1}\right) \\
\varphi_{x-1}
\end{array}\right),\quad x \in \mathbb{Z}.
\end{array}
\end{align*}
Thus the function $x \mapsto [\nu(\Psi)](x)$ is a stationary measure of the walk $\mathcal{U}_{G^{(\phi)}}$. On the other hand,
direct calculation yields
\begin{equation*}
[\nu(\Psi)](x)
=\frac{5}{4}\left(|\varphi_{x}|^{2}+|\varphi_{x-1}|^{2}\right)+\frac{1}{2} \Re\left(\varphi_{x} \overline{\varphi_{x-1}}\right),\quad
x\in \mathbb{Z}.
\end{equation*}
Thus the function $\mu$ defined by (\ref{eq-3-2}) is a stationary measure of the walk $\mathcal{U}_{G^{(\phi)}}$.
\end{proof}

\begin{theorem}\label{thr-3-2}
Consider the walk $\mathcal{U}_{\mathfrak{U}^{(\phi)}}$. Let the parameter $\phi$ be such that $\phi\in[0, 2\pi)$ and $\cos \phi \ne 0$.
Then, for each pair $\varphi_1$, $\varphi_3 \in \mathbb{C}$ with $|\varphi_1| + |\varphi_3| > 0$,
the walk $\mathcal{U}_{\mathfrak{U}^{(\phi)}}$ has a corresponding stationary measure of the form
\begin{equation}\label{eq-3-3}
  \mu(x)=\left(2+\frac{9}{4} \tan ^{2} \phi\right)\left(|\varphi_{1}|^{2}+|\varphi_{3}|^{2}\right)
       +\left(2+\frac{9}{2} \tan ^{2} \phi\right)\Re \left(\Delta(x) e^{2 i(\phi+\tau) x} \varphi_{1} \bar{\varphi}_{3}\right),\quad x\in \mathbb{Z},
\end{equation}
where $\tau\in [0,2\pi)$ is a real number such that $e^{i\tau}$ is an eigenvalue of the matrix $G^{(\phi)}$ and
\begin{equation*}
  \Delta(x)=
   \left\{
     \begin{array}{ll}
       e^{2\pi\theta i}, & \hbox{$x=-1$;} \\
       e^{-2\pi\theta i}, & \hbox{$x=1$;} \\
       1, & \hbox{$x\in \mathbb{Z}\setminus \{-1,1\}$.}
     \end{array}
   \right.
\end{equation*}
\end{theorem}

\begin{proof}
For $x\in \mathbb{Z}$, in the same way as (\ref{eq-2-2}), we have a decomposition $U_x^{(\phi)} = P_x + R_x + Q_x$.
Thus the eigenvalue problem $\mathcal{U}_{\mathfrak{U}^{(\phi)}}\Psi = \lambda\Psi$ is equivalent to
finding $\lambda\in \mathbb{C}$ and $\Psi\in l^2(\mathbb{Z},\mathbb{C}^3)$ with $\Psi\ne 0$ such that
\begin{equation*}
  \lambda \Psi(x)=P_{x-1} \Psi(x-1)+R_{x} \Psi(x)+ Q_{x+1} \Psi(x+1),\quad x\in \mathbb{Z}.
\end{equation*}
Then, by using the method of reduced matrix introduced in \cite{kkk} as well as Lemma~\ref{lem-2-2}, we can get an eigenvector $\Psi$ of $\mathcal{U}_{\mathfrak{U}^{(\phi)}}$,
which has a representation of the following form:
\begin{enumerate}
  \item[(1)] for $x\in \mathbb{Z}\setminus \{-1,0,1\}$, we have $\lambda=e^{i \tau}$, $\tilde{a}_{1}=\tilde{a}_{2}=-e^{-i \phi}$ and
$$
\Psi(x)=\left(\begin{array}{c}
\left(-e^{i \phi} e^{i\tau}\right)^{x}\varphi_{1} \\
-\left(1-\frac{3}{2} \tan \phi \cdot i \right)\left[\left(-e^{i \phi} e^{i \tau}\right)^{x} \varphi_{1}+\left(-e^{-i \phi} e^{-i \tau}\right)^{x} \varphi_{3}\right] \\
\left(-e^{-i \phi} e^{-i \tau}\right)^{x} \varphi_{3}
\end{array}\right);
$$
  \item[(2)] for $x=1$, we have $\lambda=e^{i \tau}$, $\tilde{a}_{1}=-e^{-i \phi}$, $\widetilde{a}_{2}=-\eta e^{-i \phi}$ with $\eta=e^{i2\pi \theta}$ and
$$
\Psi(x)=\left(\begin{array}{c}
\left(-e^{i \phi} e^{i \tau}\right) \varphi_{1} \\
-\left(1-\frac{3}{2} \tan \phi \cdot i\right)\left\{\left(-e^{i \phi} e^{i \tau}\right) \varphi_{1}+\left(-\eta e^{-i \phi} e^{-i \tau}\right) \varphi_{3}\right\} \\
\left(-\eta e^{-i \phi} e^{-i \tau}\right) \varphi_{3} .
\end{array}\right);
$$
  \item[(3)] for $x=-1$, we have $\lambda=e^{i \tau}$, $\tilde{a}_{1}=-\eta e^{-i \phi}$ with $\eta=e^{i2\pi \theta}$, $\widetilde{a}_{2}=-e^{-i \phi}$ and
$$
\Psi(x)=\left(\begin{array}{c}
\left(-\eta e^{-i \phi} e^{-i \tau}\right) \varphi_{1} \\
-\left(1-\frac{3}{2} \tan \phi \cdot i\right)\left\{\left(-\eta e^{-i \phi} e^{-i \tau}\right) \varphi_{1}+\left(- e^{i \phi} e^{i \tau}\right) \varphi_{3}\right\} \\
\left(- e^{i \phi} e^{i \tau}\right) \varphi_{3}
\end{array}\right);
$$
  \item[(4)] finally for $x=0$, we have $\lambda=\eta e^{i \tau}$ with $\eta=e^{i2\pi \theta}$, $\tilde{a}_{1}=\tilde{a}_{2}=-e^{-i \phi}$ and
$$
\Psi(x)=\left(\begin{array}{c}
\varphi_{1} \\
-\left(1-\frac{3}{2} \tan \phi \cdot i\right)\left(\varphi_{1}+\varphi_{3}\right) \\
\varphi_{3}
\end{array}\right).
$$
\end{enumerate}
Thus $\mu(x):= [\nu(\Psi)](x)$, $x\in \mathbb{Z}$ is a stationary measure of the walk $\mathcal{U}_{\mathfrak{U}^{(\phi)}}$. Careful calculation gives
formula (\ref{eq-3-3}).
\end{proof}

\subsection{Model II}

In this subsection we consider another model of space-inhomogeneous three-state quantum walk on the line.
Consider the unitary matrices $A_x^{(\gamma)} = e^{iv_x}A^{(\gamma)}$, $x\in \mathbb{Z}$, where $v_x$ is defined by
\begin{equation*}
  v_{x}=
\left\{
  \begin{array}{ll}
    2 \pi \theta, & \hbox{$x=0$;} \\
    0, & \hbox{$x\ne 0$, $x\in \mathbb{Z}$}
  \end{array}
\right.
\end{equation*}
with $\theta\in(0,1)$ and $A^{(\gamma)}$ is given by
\begin{equation}\label{eq}
  A^{(\gamma)}=\frac{1}{6}\left(\begin{array}{ccc}
-1-e^{2 i \gamma} & 2\left(1+e^{2 i \gamma}\right) & 5-e^{2 i \gamma} \\
2\left(1+e^{2 i \gamma}\right) & 2\left(1-2 e^{2 i \gamma}\right) & 2\left(1+e^{2 i \gamma}\right) \\
5-e^{2 i \gamma} & 2\left(1+e^{2 i \gamma}\right) & -1-e^{2 i \gamma}
\end{array}\right),
\end{equation}
where $\gamma\in[0,2\pi)$ is a parameter. Clearly, $A^{(0)} = G$, namely the Grover matrix $G$ is a special case of $ A^{(\gamma)}$.

According to Definitions~\ref{def-2-1} and \ref{def-2-2}, the family $\mathfrak{A}^{(\gamma)}:=\big\{A_x^{(\gamma)} = e^{iv_x}A^{(\gamma)}\mid x\in \mathbb{Z}\big\}$
of unitary matrices determines a space-inhomogeneous three-state quantum walk on $\mathbb{Z}$ and its evolution operator is the unitary operator
$\mathcal{U}_{\mathfrak{A}^{(\gamma)}}$ on $l^2(\mathbb{Z},\mathbb{C}^3)$ determined by the family $\mathfrak{U}^{(\phi)}$.
In what follows, we simply call this model the walk $\mathcal{U}_{\mathfrak{A}^{(\gamma)}}$.

Note that $A^{(\gamma)}$ is also a unitary matrix, hence it determines a space-homogeneous three-state quantum walk on $\mathbb{Z}$,
whose evolution operator is the unitary operator $\mathcal{U}_{A^{(\gamma)}}$ determined by $A^{(\gamma)}$.
Similarly, we call the space-homogeneous three-state quantum walk determined by $A^{(\gamma)}$ the walk $\mathcal{U}_{A^{(\gamma)}}$ below.

\begin{remark}\label{rem-3-2}
The walk $\mathcal{U}_{\mathfrak{A}^{(\gamma)}}$ includes those considered in \cite{wlw, ekk-2} as a special case.
\end{remark}

In 2017, Kawai et al. \cite{kkk} obtained stationary measures of the walk $\mathcal{U}_{A^{(\gamma)}}$, which is space-homogeneous as mentioned above.
In the following, we would like to find out stationary measures of the  space-inhomogeneous walk $\mathcal{U}_{\mathfrak{A}^{(\gamma)}}$.

\begin{theorem}
Consider the walk $\mathcal{U}_{\mathfrak{A}^{(\gamma)}}$. Let $\varphi_1$, $\varphi_3 \in \mathbb{C}$ be such that $|\varphi_1| + |\varphi_3| > 0$.
Then the walk $\mathcal{U}_{\mathfrak{A}^{(\gamma)}}$  has a stationary measure $\mu$ given by
\begin{equation}\label{eq-3-5}
\mu(x)=\left(2+\frac{9}{4} \tan ^{2} \gamma\right)\left(|\varphi_{1}|^{2}+|\varphi_{3}|^{2}\right)
+ \left(2+\frac{9}{2} \tan ^{2} \gamma\right)\Re \left(\Delta(x) e^{2 i\xi x} \varphi_{1} \bar{\varphi}_{3}\right),\quad x\in \mathbb{Z},
\end{equation}
where $\xi \in [0,2\pi)$ such that $e^{i \xi} = \frac{10-26 \cos (2 \gamma)-24 i \sin (2 \gamma)}{26-10 \cos (2 \gamma)}$, which is an eigenvalue of $A^{(\gamma)}$,
and $\Delta(\cdot)$ is the function on $\mathbb{Z}$ defined by $\Delta(-1)=e^{2\pi \theta i}$, $\Delta(1)=e^{-2\pi \theta i}$ and $\Delta(x)=1$ for $x\in \mathbb{Z}\setminus\{-1,1\}$.
\end{theorem}

\begin{proof}
For $x\in \mathbb{Z}$, in the same way as (\ref{eq-2-2}), we have a decomposition $A_x^{(\gamma)} = P_x + R_x + Q_x$.
Similar to that in the proof of Theorem~\ref{thr-3-2}, we need to find $\lambda\in \mathbb{C}$ and $\Psi\in l^2(\mathbb{Z},\mathbb{C}^3)$ with $\Psi\ne 0$ such that
\begin{equation*}
  \lambda \Psi(x)=P_{x-1} \Psi(x-1)+R_{x} \Psi(x)+ Q_{x+1} \Psi(x+1),\quad x\in \mathbb{Z}.
\end{equation*}
By using the method of reduced matrix introduced in \cite{kkk} as well as Lemma~\ref{lem-2-1}, we can such a function $\Psi$. The function $\Psi$
has a representation of the following form:
\begin{enumerate}
  \item[(1)] for $x \in \mathbb{Z}\setminus \{-1, 0, 1\}$, by Lemma~\ref{lem-2-2}, we have $\lambda=e^{i\xi}$, $\tilde{a}_{1}=\tilde{a}_{2}=-1$ and
$$
\Psi(x)=\left(\begin{array}{c}
-e^{i \xi x} \varphi_{1} \\
-\left(1-\frac{3}{2} \tan \gamma \cdot i\right)\left[-e^{i \xi x} \varphi_{1}-e^{-i \xi x} \varphi_{3}\right] \\
-e^{-i \xi x} \varphi_{3}
\end{array}\right);
$$
  \item[(2)] for $x=1$, we have $\lambda=e^{i\xi}$, $\tilde{a}_{1}=-1$, $\tilde{a}_{2}=-\eta$ with $\eta=e^{2\pi\theta i}$ and
$$
\Psi(x)=\left(\begin{array}{c}
-e^{i \xi } \varphi_{1} \\
-\left(1-\frac{3}{2} \tan \gamma \cdot i\right)\left[-e^{i \xi } \varphi_{1}-\eta e^{-i \xi } \varphi_{3}\right] \\
-\eta e^{-i \xi } \varphi_{3}
\end{array}\right);
$$
  \item[(3)] for $x=-1$, we have $\lambda=e^{i\xi}$, $\tilde{a}_{1}=-\eta$ with $\eta=e^{2\pi\theta i}$, $\tilde{a}_{2}=-1$ and
$$
\Psi(x)=\left(\begin{array}{c}
-\eta e^{-i \xi } \varphi_{1} \\
-\left(1-\frac{3}{2} \tan \gamma \cdot i\right)\left[-\eta e^{-i \xi } \varphi_{1}- e^{i \xi } \varphi_{3}\right] \\
-e^{ i \xi } \varphi_{3}
\end{array}\right),
$$
  \item[(4)] for $x=0$, we have $\lambda=\eta e^{i \xi}$ with $\eta=e^{2\pi\theta i}$, $\tilde{a}_{1}=\tilde{a}_{2}=-1$ and
$$
\Psi(x)=\left(\begin{array}{c}
\varphi_{1} \\
-\left(1-\frac{3}{2} \tan \gamma \cdot i\right)\left(\varphi_{1}+\varphi_{3}\right) \\
\varphi_{3}
\end{array}\right).
$$
\end{enumerate}
Thus the function $\mu(x):= [\nu(\Psi)](x)$, $x\in \mathbb{Z}$ is a stationary measure of the walk $\mathcal{U}_{\mathfrak{A}^{(\gamma)}}$.
Careful calculations then yield (\ref{eq-3-5}).
\end{proof}

\section{Conclusion remarks}\label{sec-4}

Although the method of reduced matrix \cite{kkk} was originally developed for dealing with space-homogeneous three-state quantum walks (QWs) on the line,
our work in this paper shows that this method can also play a role in dealing with some space-inhomogeneous three-state QWs.
Recently, Endo et al. \cite{ekkt} have developed a method, known as the method of transfer matrix, which can deal with some more general
space-inhomogeneous three-state QWs. However, to deal with a general space-inhomogeneous three-state QW, a more powerful method is needed.
It is still challenging to make clear the whole picture of the set of stationary measures of a general space-inhomogeneous three-state QW.

\section*{Acknowledgement}

This work is supported by National Natural Science Foundation of China (Grant No. 12261080).

\end{document}